\documentclass[12pt,twoside]{article}

 \usepackage{float}
\usepackage{graphicx}
\usepackage{epstopdf}
\usepackage{graphicx}
\usepackage{epic}
\usepackage{multirow}
\usepackage{tikz}
\usepackage{xcolor}
\usepackage{threeparttable}
\usetikzlibrary{arrows,shapes,chains}

\usepackage{graphicx}
\usepackage{makecell}
\renewcommand{\paragraph}{\roman{paragraph}}
\usepackage[a4paper]{geometry}
\makeatletter
\renewcommand\title[1]{\gdef\@title{\reset@font\Large\bfseries #1}}
\renewcommand\section{\@startsection {section}{1}{\z@}%
                                   {-3.5ex \@plus -1ex \@minus -.2ex}%
                                   {2.3ex \@plus.2ex}%
                                   {\normalfont\large\bfseries}}
\renewcommand\subsection{\@startsection{subsection}{2}{\z@}%
                                     {-3ex\@plus -1ex \@minus -.2ex}%
                                     {1.5ex \@plus .2ex}%
                                     {\normalfont\normalsize\bfseries}}
\renewcommand\subsubsection{\@startsection{subsubsection}{3}{\z@}%
                                     {-2.5ex\@plus -1ex \@minus -.2ex}%
                                     {1.5ex \@plus .2ex}%
                                     {\normalfont\normalsize\bfseries}}

\def\@runningauthor{}\newcommand{\runningauthor}[1]{\def\runningauthor{#1}}
\def\@runningtitle{}\newcommand{\runningtitle}[1]{\def\runningtitle{#1}}

\renewcommand{\ps@plain}{%
\renewcommand{\@evenhead}{\footnotesize\scshape \hfill\runningauthor\hfill}
\renewcommand{\@oddhead}{\footnotesize\scshape \hfill\runningtitle\hfill}}
\newcommand{\Z}{\mathbb{Z}}
\newcommand{\F}{\mathbb{F}}

\g@addto@macro\bfseries{\boldmath}

\makeatother

\usepackage{amsthm,amsmath,amssymb}
\usepackage{cite}
\usepackage{graphicx}

\usepackage[colorlinks=true,citecolor=black,linkcolor=black,urlcolor=blue]{hyperref}

\theoremstyle{plain}
\newtheorem{theorem}{Theorem}[section]

\newtheorem{lemma}[theorem]{Lemma}
\newtheorem{cor}[theorem]{Corollary}
\newtheorem{prop}[theorem]{Proposition}

\theoremstyle{definition}
\newtheorem{definition}[theorem]{Definition}
\newtheorem{example}[theorem]{Example}

\theoremstyle{remark}
\newtheorem{remark}[theorem]{Remark}

\runningauthor{}

\date{}

\begin{document}
\title{Characterization of Plotkin-optimal two-weight codes over finite chain rings and related applications\thanks{This research is supported by the National Natural Science Foundation of China (12071001 and 12201170) and the Natural Science Foundation of Anhui Province (2108085QA03).}}
\author{Shitao Li\thanks{Shitao Li is with the School of Mathematical Sciences, Anhui University, Hefei, 230601, China, email: lishitao0216@163.com.},
Minjia Shi\thanks{Minjia Shi is with the Key Laboratory of Intelligent Computing and Signal Processing, Ministry of Education, School of Mathematical Sciences, Anhui University, Hefei, 230601, China, and also the State Key Laboratory of Integrated Services Networks, Xidian University, Xi'an, 710071, China, email: smjwcl.good@163.com.}
}
\date{}
    \maketitle
\begin{abstract}
Few-weight codes over finite chain rings are associated with combinatorial objects such as strongly regular graphs (SRGs), strongly walk-regular graphs (SWRGs) and finite geometries, and are also widely used in data storage systems and secret sharing schemes. The first objective of this paper is to characterize all possible parameters of Plotkin-optimal two-homogeneous weight regular projective codes over finite chain rings, as well as their weight distributions. We show the existence of codes with these parameters by constructing an infinite family of two-homogeneous weight codes. The parameters of their Gray images have the same weight distribution as that of the two-weight codes of type SU1 in the sense of Calderbank and Kantor (Bull Lond Math Soc 18: 97-122, 1986). Further, we also construct three-homogeneous weight regular projective codes over finite chain rings combined with some known results. Finally, we study applications of our constructed codes in secret sharing schemes and graph theory. In particular, infinite families of SRGs and SWRGs with non-trivial parameters are obtained.
\end{abstract}


\section{Introduction}
Since the 1970s, a one-to-one correspondence between two-weight projective codes over prime fields and some strongly regular graphs (SRGs) was discovered by Delsarte \cite{DM-Delsarte}. In 2008, Byrne, Greferath and Honold \cite{DCC-2-SRGs} extended
this classical result to finite Frobenius rings, and showed that a regular projective code over a finite Frobenius ring with two nonzero homogeneous weights also determined an SRG. Recently, Shi $et~ al.$ \cite{IT-01,Shi-S-Sole} studied the relations between the dual codes of two-homogeneous weight regular projective codes and SRGs over $\Z_{2^k}$ and $\Z_{p^m}$. For more connections between two-weight codes and SRGs, one can refer to \cite{FFA-2-weight}. Some SRGs were constructed from regular projective two-weight codes over finite chain rings \cite{SWP,Tang-DCC,JAC-Shi,FFA-zhu-liao}.

Strongly walk-regular graphs (SWRGs) were introduced in \cite{JCTA-SWRG} as a natural generalization of SRGs, where paths of length $2$ are replaced by paths of length $s$. Shi and Sol\'e \cite{DCC-shi-sole} connected together the two notions of triple sum sets (TSSs) and SWRGs over finite fields, and constructed many infinite families of SWRGs as coset graphs of the dual codes of three-weight projective codes over finite fields whose weights satisfy a certain equation. This led Kiermaier $et~al.$ to classify the feasible parameters of these codes in the binary and ternary case for medium size code lengths \cite{DCC-Kiermaier}. Recently, the authors \cite{Shi-GC} extended these results to finite chain rings, and classified short length three-weight codes leading to SWRGs over $\Z_4$ and $\F_2+u\F_2$ with $u^2=0$.

Besides some connections between few-weight codes and combinatorial structures, few-weight codes also have applications in data storage systems and secret sharing schemes \cite{SS-Ding,ass-Calderbank,ACM-Shamir,SSS-Blakley}. Recently, there were some progress in the study of two-weight codes over finite rings \cite{BAMS-Tang,li,SW,Shi-Xu-Yang,Tang-DCC,SWP,Liu-IT}. Particularly, Shi and Wang \cite{SW} studied the properties and constructions of two-Lee weight projective codes over $\Z_4$, and their linearity was analyzed in \cite{Shi-Xu-Yang}. Later, Shi $et~al.$ \cite{SWP} constructed two infinite families of two-Lee weight codes over $\Z_4$ by their generator matrices, which generalized the results in \cite{SW}. Note that all two-Lee weight codes over $\Z_4$ constructed in \cite{SW,Shi-Xu-Yang,SWP} are Plotkin-optimal, regular and projective. More recently, an interesting result is that Tang and Suprijanto \cite{Tang-DCC} characterized possible parameters of Plotkin-optimal two-Lee weight regular projective codes over $\Z_4$, together with their weight distributions. They also showed the existence of codes with these parameters by unifying the construction method in \cite{SWP}.

Inspired and motivated by the works listed above, we replace the ring $\Z_4$ by a finite chain ring. The purpose of this paper is to characterize possible parameters and weight distributions of Plotkin-optimal two-homogeneous weight regular projective codes over finite chain rings. We show the existence of codes with these parameters by a generalized construction method. Their Gray images have the same weight distribution as that of the two-weight codes of type SU1 in the sense of Calderbank and Kantor \cite{CK}. Thus they produce the same SRGs as those given by the SU1 codes, but described as Cayley graphs on a different abelian group. As a by-product, we construct several families of three-homogeneous weight regular projective codes over finite chain rings combined with some known results. As an application, infinite families of SWRGs are obtained.

The paper is organized as follows. In Section 2, we give some notations and definitions. In Section 3, we study some properties of Plotkin-optimal regular projective codes over finite chain rings. In Section 4, we focus on Plotkin-optimal regular projective two-homogeneous weight codes over finite chain rings. In Section 5, we study applications of our constructed codes in secret sharing schemes and graph
theory. In Section 6, we conclude the paper.

\section{Preliminaries}
\subsection{Finite chain rings}
We begin with the definition and some properties of finite chain rings following mainly \cite{finite-ring}.
\begin{definition}
A finite commutative ring with $1\neq 0$ is called a {\em finite chain ring} if its ideals are linearly ordered by inclusion.
\end{definition}
A simple example of a finite chain ring is the ring $\Z_{p^m}$ of integers modulo $p^m$, for some prime $p$ and integer $m\geq 1$.
It is well-known that if $R$ is a finite chain ring, then $R$
is a principal ideal ring and it has a unique maximal ideal $\langle \theta\rangle$. Its chain of ideals is
$$\{0\}=\langle \theta^m\rangle\subsetneq \langle \theta^{m-1}\rangle\subsetneq\cdots \subsetneq \langle \theta\rangle \subsetneq R,$$
where $\langle \theta^i\rangle$ is the ideal of $R$ generated by $\theta^i$ for $1\leq i\leq m.$ The integer $m$ is called the {\em depth} of $R$. It follows that $R/ \langle \theta\rangle$ is a finite field, which is called the {\em residue field} of $R$ and we denote it by $\F_q$. Let $R^\times$ denote the set of all the units in $R$. Then $R^\times=R\backslash \langle \theta\rangle$. It follows that $|R|=q^m$ and $|R^\times|=q^m-q^{m-1}$. Moreover,
$|\langle \theta^i\rangle|=q^{m-i}$ for $1\leq i\leq m.$

Note that there is the natural surjection $\nu:~R\longrightarrow \F_q$. Define the {\em Teichm$\ddot{ u}$ller set} $T$ as a set of $q$ elements of $R$ such that $T \equiv \F_q~(\bmod ~\langle \theta\rangle)$. Then there is a unique $(a_0,a_1,\ldots,a_{m-1})\in T^m$ such that the {\em $\theta$-adic expansion} of an arbitrary element $a\in R$ is
$$a=a_0+a_1\theta+\cdots+a_{m-1}\theta^{m-1}.$$
We call $a^{(i)}:=\nu(a_i)\in \F_q$ the {\em $\theta$-adic components} of $a$.
It is easy to see that $a\in R^\times$ if and only if $a_0\neq 0$.

\subsection{Linear codes over finite chain rings}
Let $R$ be a finite chain ring of the depth $m$ with the residue field $\F_q$ and the maximum idea $ \langle \theta\rangle$.
A {\em linear code} over $R$ of length $n$ is an $R$-submodule of $R^n$. A matrix $G$ is called a {\em generator matrix} of the linear code $C$ if the rows of $G$ generate $C$ as an $R$-module. A matrix $H$ is called a {\em parity-check matrix} of the linear code $C$ if $H$ is a generator matrix of $C^\perp$. By \cite{AAECC}, any linear code $C$ of length $n$ over $R$ is permutation equivalent to a linear code over $R$ with the following generator matrix
$$
    G=\begin{bmatrix}
    \begin{array}{cccccc}
    I_{k_0} & A_{0,1} & A_{0,2} & \cdots & A_{0,m-1} & A_{0,m}\\
    0 & \theta I_{k_1} & \theta A_{1,2} & \cdots & \theta A_{1,m-1} & \theta A_{1,m}\\
    0 & 0 & \theta^2I_{k_2} & \cdots & \theta^2A_{2,m-1} & \theta^2A_{2,m}\\
    \cdots & \cdots& \cdots & \cdots & \cdots & \cdots\\
    0 & 0 & 0 & \cdots & \theta^{m-1}I_{k_{m-1}} & \theta^{m-1}A_{m-1,m}
    \end{array}
    \end{bmatrix},
    $$
where $I_{k_i}$ is the identity matrix of order $k_i$, $A_{i,j}$ is a matrix over $R$ and the columns are grouped into blocks of sizes $k_0,\ldots,k_{m-1},n-\sum_{i=0}^{m-1}k_i$, respectively. Then $|C|=q^k$, where $k=\sum_{i=0}^{m-1}(m-i)k_i$ is the {\em $q$-dimension} of $C$, and we say $C$ has {\em type} $1^{k_0}q^{k_1}\cdots (q^{m-1})^{k_{m-1}}$. A code of type $1^{k_0}q^{k_1}\cdots (q^{m-1})^{k_{m-1}}$ over $R$ of length $n$ is denoted as an $[n;k_0,k_1,\ldots,k_{m-1}]_q$ code. A linear $[n;k_0,k_1,\ldots,k_{m-1}]_q$ code over $R$ is {\em free} if and only if $k_1=k_2=\cdots=k_{m-1}=0.$ 

\begin{definition}
Let $R$ be a finite ring. A function $\omega:~R\rightarrow \mathbb{R}$ is called a {\em homogeneous weight}, if $\omega(0)=0$ and
\begin{itemize}
  \item [{\rm (1)}] if $Rx=Ry$ then $\omega(x)=\omega(y)$ for all $x,y\in R$;
  \item [{\rm (2)}] there exists a real number $\gamma\neq 0$ such that $\sum_{y\in Rx}\omega(y)=\gamma |Rx|$ for all $x\in R \backslash \{0\}.$
\end{itemize}
\end{definition}
By \cite{C}, a homogeneous weight exists for all finite commutative rings, and is uniquely determined up to the normalization factor $\gamma$.
In the sequel we will denote the homogeneous weight by ${\rm wt_{hom}}$. If the chosen normalization constant $\gamma=(q-1)q^{m-2}$, then the homogeneous weight of an element of $R$ is defined as
$${\rm wt_{hom}}(x)=\left\{
\begin{array}{ll}
0,& {\rm if}\ x=0, \\
q^{m-1}, &{\rm if}\ x\in \langle \theta^{m-1}\rangle\backslash \{0\},\\
(q-1)q^{m-2},&{\rm if} ~x \in R\backslash \langle \theta^{m-1}\rangle.
\end{array}
\right.
$$
In particular, for $R=\Z_4$, the homogeneous weight coincides with the classical Lee
weight of \cite{HKCSS}.
For ${\bf x}=(x_1,\ldots,x_n)\in R^n$, its homogeneous weight is ${\rm wt_{ hom}}({\bf x})=\sum_{i=1}^n {\rm wt_{hom}}(x_i)$.
The {\em homogeneous distance} between two vectors ${\bf x}$ and ${\bf y}$ is ${\rm d_{hom}}({\bf x},{\bf y})={\rm wt_{hom}({\bf x}-{\bf y})}$.
The {\em minimum homogeneous distance} of a linear code $C$ over $R$ is the smallest homogeneous distance between distinct codewords of $C$. Let $(A_0,A_1,\ldots, A_{nq^{m-1}})$ be the homogeneous weight distribution of a code $C$, where $A_i = \{{\bf c} \in C~|~{\rm wt_{hom}}({\bf c}) = i\}$.
For $\textbf x=(x_1,x_2,\cdots,x_n)$ and $\textbf y=(y_1,y_2,\cdots,y_n)\in R^n$, the {\em inner product} of ${\bf x}$ and ${\bf y}$ is defined as
$$\textbf x \cdot \textbf y=\sum _{i=1}^n x_iy_i.$$
The {\em dual code} of a linear code $C$ of length $n$ over $R$ is defined as
$$C^{\perp} =\{\textbf y\in R^n ~|~ \textbf x \cdot \textbf y=0, \forall\ \textbf x\in C \}.$$
An {\em $N$-homogeneous weight code} is a code such that the cardinality of the set of nonzero homogeneous weights is $N$.

\begin{definition}
Let $C$ denote a linear code over $R$ of length $n$ with the $\ell \times n$ generator matrix $G=[{\bf g}_1~|~{\bf g}_2~|~\cdots~|~{\bf g}_n]$. The code $C$ is called
\begin{itemize}
  \item [(1)] {\em Proper}: if ${\rm wt_{hom}}({\bf c}) = 0$ implies ${\bf c} = {\bf 0}$ for all ${\bf c} \in C$;
  \item [(2)] {\em Regular}: if $\{{\bf x}\cdot {\bf g}_i~|~{\bf x}\in R^\ell\}=R$ for each $i=1,2,\ldots,n$;
  \item [(3)] {\em Projective}: if ${\bf g}_iR\neq {\bf g}_jR$ for any pair of distinct coordinates $i,j\in\{1,2,\ldots,n\}$.
\end{itemize}
\end{definition}

\begin{remark}
Since ${\rm wt_{hom}}( x)>0$ for all $x\in R\backslash\{0\}$, all linear codes over $R$ are proper. The regularity reduces to the condition that every column of a generator matrix contains at least one unit. For a $\Z_4$-code to be both regular and projective is equivalent to $d^\perp\geq 3$, where $d^\perp$ denotes the dual Lee distance.
\end{remark}

\subsection{The Gray map for the homogeneous weight} Let $R$ be a finite chain ring of the depth $m$ with the residue field $\F_q$ and the maximum idea $ \langle \theta\rangle$. Next, we recall the Gray map for the homogeneous weight over the finite chain ring $R$ \cite{GS-IT}. Let ${\bf u}$ and ${\bf v}$ be two vectors of $\F_q^q$ such that ${\bf u}$ lists all elements of $\F_q$ whereas ${\bf v}$ is the all-one vector. Let
$${\bf c}_i=({\bf v}+\delta_{i,0}({\bf u-v}))\otimes ({\bf v}+\delta_{i,1}({\bf u-v}))\otimes ({\bf v}+\delta_{i,m-2}({\bf u-v}))
$$
for $i=0,1,\ldots,m-1$, where $\delta_{i,j}$ denotes the Kronecker symbol and $\otimes$ is the tensor product (expanded from right to left) over $\F_q$. For example, set
${\bf u}=(u_1,\ldots,u_q)$, we have ${\bf u}\otimes {\bf v}=(u_1{\bf v},\ldots,u_q{\bf v}).$
Then ${\bf v}_1,{\bf v}_2,\ldots,{\bf v}_{m-1}$ span the first-order Reed-Muller $[q^{m-1},m,(q-1)q^{m-2}]$ code over $\F_q$. The {\em Gray map} $\phi$ is defined by
\begin{align*}
  \phi: R \longrightarrow& ~\F_q^{q^{m-1}}, \\
  a \longmapsto& ~a^{(0)}{\bf c}_0+a^{(1)}{\bf c}_1+\cdots+a^{(m-1)}{\bf c}_{m-1}.
\end{align*}
This map can be extended to $\Phi:~R^n\longrightarrow~\F_q^{q^{m-1}n}$ naturally. It has been shown \cite{GS-IT} that $\Phi$ is a distance-preserving map from ($R^n$, ${\rm d_{hom}}$) to ($\F_q^{q^{m-1}n}$, ${\rm d_H}$), where ${\rm d_H}$ denotes the Hamming distance. The definition of Hamming distance refers to \cite{Huffman}. An $(n,M,d)$ code $C$ over $\F_q$ is a set of $\F_q^n$ such that $|C|=M$ and $C$ has the minimum Hamming distance $d$.
In particular, for $R=\Z_{2^m}$, the Gray map is the same as Carlet's Gray map \cite{IT-Carlet}.

\section{Plotkin-optimal regular projective codes over finite chain rings}

In this section, let $R$ be a finite chain ring of the depth $m$ with the residue field $\F_q$ and the maximum idea $ \langle \theta\rangle$. Then $R^*=R\backslash \langle\theta\rangle$ and it follows that $|R^*|=q^m-q^{m-1}$.

\begin{theorem}(Plotkin-type homogeneous distance bound)
If $C$ is a linear $[n;k_0,k_1,\ldots,k_{m-1}]_q$ code with the minimum homogeneous distance ${\rm d_{hom}}(C)$ over $R$, then
$${\rm d_{hom}}(C)\leq \left\lfloor\frac{(q^{m-1}-q^{m-2})n|C|}{|C|-1}\right\rfloor.$$
\end{theorem}

\begin{proof}
By Subsection 2.3, $\Phi(C)$ is a $(q^{m-1}n,|C|,{\rm d_{hom}}(C))$ code over $\F_q$. If ${\rm d_{hom}(C)}\leq (1-\frac{1}{q})n$, then it can be checked that
$$\left(1-\frac{1}{q}\right)n\leq \left\lfloor\frac{(q^{m-1}-q^{m-2})n|C|}{|C|-1}\right\rfloor.$$
Hence the result follows.
If ${\rm d_{hom}(C)}> (1-\frac{1}{q})n$, then it follows from {\rm \cite[Theorem 2.2.1]{Huffman}} that $$|C|\leq \left\lfloor\frac{{\rm d_{hom}}(C)}{{\rm d_{hom}}(C)-(1-\frac{1}{q})n}\right\rfloor.$$
We can obtain the desired result after simplification.
\end{proof}

\begin{remark}
If $R=\Z_{2^m}$, then Plotkin-type homogeneous distance bound is the same as the bound in \cite[Corollary 1]{G-B-L}.
A linear code $C$ over $R$ is said to be {\em Plotkin-optimal} if $C$ meets the Plotkin-type homogeneous distance bound, that is, ${\rm d_{hom}}(C)= \left\lfloor\frac{(q^{m-1}-q^{m-2})n|C|}{|C|-1}\right\rfloor.$
\end{remark}

For convenience, we introduce a notation. Let $G$ be a $k\times n$ matrix over $R$. For any ${\bf c}\in R^k$, the {\em multiplicity} $\mu_G({\bf c})$ of ${\bf c}$ in $G$ is defined by Kl${\rm {\o}}$ve \cite{DM-Klove} as the number of occurrences of ${\bf c}$ as a column vector in $G$.

\begin{lemma}\label{lemma-n}
If $C$ is a regular projective linear $[n;k_0,k_1,\ldots,k_{m-1}]_q$ code over $R$, then
$k_0\geq 1$ and
$$n\leq \frac{q^{k}-q^{k-k_0}}{q^m-q^{m-1}},~i.e.,~(q^m-q^{m-1})n\leq q^{k}-q^{k-k_0},$$
where $k=mk_0+(m-1)k_1+\cdots+k_{m-1}$ is the $q$-dimension of $C$.
\end{lemma}

\begin{proof}
Let $k=mk_0+(m-1)k_1+\cdots+k_{m-1}$. Let
$$S_1=\langle \theta\rangle^{k_0}\times \langle \theta\rangle^{k_1}\times\cdots\times \langle \theta^{m-1}\rangle^{k_{m-1}}$$
and
$$S_2=(R^{k_0}\backslash \langle \theta\rangle^{k_0})\times \langle \theta\rangle^{k_1}\times\cdots\times \langle \theta^{m-1}\rangle^{k_{m-1}}.$$
Suppose that $G$ is the generator matrix of $C$. Since $C$ is regular, $k_0\geq 1$ and $\mu_G({\bf c})=0$ for ${\bf c}\in S_1$. Since $C$ is projective, $\sum_{a\in R^\times}\mu_G(a{\bf c}) \leq 1$ for ${\bf c}\in S_2$.
Hence
$$n=\sum_{{\bf c}\in S_1\cup S_2}\mu_G({\bf c})=\sum_{{\bf c}\in S_1}\mu_G({\bf c})+
\sum_{{\bf c}\in S_2}\mu_G({\bf c})\leq \frac{|S_2|}{q^{m}-q^{m-1}}=\frac{q^k-q^{k-k_0}}{q^m-q^{m-1}}.$$
This completes the proof.
\end{proof}

\begin{lemma}\label{lemma-P-optimal}
If $C$ is a regular projective linear $[n;k_0,k_1,\ldots,k_{m-1}]$ code over $R$, then $C$ is Plotkin-optimal if and only if ${\rm d_{ hom}}(C)=(q^{m-1}-q^{m-2})n$.
\end{lemma}

\begin{proof}
Suppose that $k=mk_0+(m-1)k_1+\cdots+k_{m-1}$. Let $C$ be a regular projective linear $[n;k_0,k_1,\ldots,k_{m-1}]$ code over $R$.
By Lemma \ref{lemma-n}, we have $k_0\geq 1$ and
$$(q^{m-1}-q^{m-2})n<(q^{m}-q^{m-1})n\leq q^k-q^{k-k_0}\leq q^k-1=|C|-1.$$
This implies that
$\left\lfloor\frac{(q^{m-1}-q^{m-2})n}{|C|-1}\right\rfloor=0$.
Hence we have
\begin{align*}
  \left\lfloor\frac{(q^{m-1}-q^{m-2})n|C|}{|C|-1}\right\rfloor & =  \left\lfloor\frac{(q^{m-1}-q^{m-2})n(|C|-1+1)}{|C|-1}\right\rfloor\\
   & =(q^{m-1}-q^{m-2})n+\left\lfloor\frac{(q^{m-1}-q^{m-2})n}{|C|-1}\right\rfloor\\
   &=(q^{m-1}-q^{m-2})n.
\end{align*}
Therefore, $C$ is Plotkin-optimal if and only if ${\rm d_{ hom}}(C)=(q^{m-1}-q^{m-2})n$.
\end{proof}

\begin{lemma}\label{A}
For any $\textbf c\in R^n$ and $0\leq j\leq m-1$, we have
$$\sum_{a\in \langle\theta^j\rangle}{\rm wt_{hom}}(\textbf c + a{\rm{\bf 1}})=(q-1)q^{2m-2-j}n,$$
where ${\bf 1}$ is the all-one row vector of length $n$.
\end{lemma}

\begin{proof}
Suppose that $c\in R$ and $0\leq j\leq m-1$. If $c \notin \langle\theta^j\rangle$, then $(c + \langle\theta^j\rangle) \cap \langle\theta^{m-1}\rangle=\emptyset$. Hence ${\rm wt_{hom}}(c+a)=(q-1)q^{m-2}$ for any $a\in \langle\theta^j\rangle$ and
$$\sum_{a\in \langle\theta^j\rangle}{\rm wt_{hom}}(c+a)=\sum_{a\in \langle\theta^j\rangle}(q-1)q^{m-2}=(q-1)q^{2m-2-j}.$$
If $c\in \langle\theta^j\rangle$, then $c + \langle\theta^j\rangle=\langle\theta^j\rangle$. So
\begin{align*}
  \sum_{a\in \langle\theta^j\rangle}{\rm wt_{hom}}(c+a) & =\sum_{a\in \langle\theta^j\rangle}{\rm wt_{hom}}(a) \\
   & =\sum_{a\in \langle\theta^{m-1}\rangle}{\rm wt_{hom}}(a)+\sum_{a\in \langle\theta^j\rangle\backslash \langle\theta^{m-1}\rangle}{\rm wt_{hom}}(a)\\
   & =(q-1) q^{m-1}+(q^{m-j}-q) (q-1)q^{m-2}\\
   &=(q-1)q^{2m-2-j}.
\end{align*}
It follows that
$$\sum_{a\in \langle\theta^j\rangle}{\rm wt_{hom}}(\textbf c + a{\rm{\bf 1}}) =
\sum_{k=1}^n\sum_{a\in \langle\theta^j\rangle}{\rm wt_{hom}}(c_k+a)=\sum_{k=1}^n(q-1)q^{2m-2-j}=(q-1)q^{2m-2-j}n.$$
This completes the proof.
\end{proof}

\begin{lemma}\label{lem-one-weight}
Let $k_i$ be a nonnegative integer, where $1\leq i\leq m$. Let $C$ be a linear $[n;k_0,k_1,\ldots,k_{m-1}]_q$ code over $R$ with the generator matrix $G$ whose columns are all distinct nonzero vectors
$$(c_1,\ldots,c_{k_1},c_{k_1+1},\ldots,c_{k_1+k_2},\ldots, c_{k_1+\cdots +k_{m-1}+1},\ldots,c_{k_1+\cdots+k_m})^{{\rm T}},$$
where $c_{i_j}\in \langle\theta^j\rangle$ for $k_1+\cdots+k_{j}+1\leq i_j\leq k_1+\cdots+k_{j+1}$ and $0\leq j\leq m-1$.
Then $C$ is a one-homogeneous weight code with the nonzero homogeneous weight $\omega=q^{k}(q^{m-1}-q^{m-2})$ and $n=q^k-1$, where $k$ is the $q$-dimension of $C$.
\end{lemma}

\begin{proof}
The result is obvious by Lemma \ref{A}.
\end{proof}

\begin{remark}
It can be checked that the Gray image of the code of Lemma \ref{lem-one-weight} is a one-Hamming weight $(q^{m-1}(q^k-1),q^k,q^k(q^{m-1}-q^{m-2}))$ code, which attains the Plotkin Hamming bound (see \cite[Theorem 2.2.1]{Huffman}).
For the study of one-weight codes over rings, one refers to \cite{one-z4,wood}. Here we give a general construction method for one-homogeneous weight codes over finite chain rings.
\end{remark}

\begin{theorem}\label{thm-1}
Let $C$ be an $M$-homogeneous weight linear $[n;k_0,k_1,\ldots,k_{m-1}]$ code with generator matrix $G$ over $R$. Suppose that $C$ has the weight distribution
$$1+A_{\omega_1}X^{\omega_1}+A_{\omega_2}X^{\omega_2}+\cdots+A_{\omega_M}X^{\omega_M}.$$
Let $C_0$ be the linear code over $R$ with the following generator matrix
$$G_0=\left[
        \begin{array}{c|c|c|c|c}
          G & G & G&\cdots & G \\
          \hline
          {\bf 0} & a_1{\bf 1}& a_2{\bf 1} & \cdots& a_{q^{m-m_0}-1}{\bf 1} \\
        \end{array}
      \right],
$$
where $0\leq m_0\leq m-1$ and $\langle \theta^{m_0}\rangle=\{0,a_1,a_2,\ldots,a_{q^{m-m_0}-1}\}$. Then the following statements hold.
\begin{itemize}
  \item [{\rm(1)}] If there exists $1\leq i_0\leq M$ such that $\omega_{i_0}=(q^{m-1}-q^{m-2})n$, then $C_0$ is an $M$-homogeneous weight linear $[q^{m-m_{0}}n;k_0,\ldots,k_{m_0-1},k_{m_0}+1,k_{m_0+1},\ldots,k_{m-1}]$ code. Moreover, $C_0$ has the homogeneous weight distribution $$1+B_{\omega'_1}X^{\omega'_1}+B_{\omega'_2}X^{\omega'_2}+\cdots+
      B_{\omega'_M}X^{\omega'_M},$$
where $\omega'_i=q^{m-m_0}\omega_i$ and $$B_{w'_i}=\left\{\begin{array}{ll}
                         A_{w_i}, & {\rm if}~i\neq i_0, \\
                         A_{w_i}+(q^{m-m_0}-1)|C|, & {\rm if}~i= i_0.
                       \end{array}
\right.$$

  \item [{\rm(2)}] If there is no $1\leq i\leq M$ such that $\omega_{i}=(q^{m-1}-q^{m-2})n$, then $C_0$ is an $(M+1)$-homogeneous weight linear $[q^{m-m_{0}}n;k_0,\ldots,k_{m_0-1},k_{m_0}+1,k_{m_0+1},\ldots,k_{m-1}]$ code. Moreover, $C_0$ has the homogeneous weight distribution
       $$1+B_{\omega'_1}X^{\omega'_1}+B_{\omega'_2}X^{\omega'_2}+\cdots+
      B_{\omega'_M}X^{\omega'_M}+B_{\omega'_{M+1}}X^{\omega'_{M+1}},$$
where
\begin{center}
{\small $\omega'_i=\left\{\begin{array}{ll}
                         q^{m-m_0}\omega_i, & {\rm if}~1\leq i\leq M, \\
                         (q-1)q^{2m-2}n, & {\rm if}~i= M+1.
                       \end{array}
\right.$ and
$B_{w'_i}=\left\{\begin{array}{ll}
                         A_{w_i}, & {\rm if}~1\leq i\leq M, \\
                         (q^{m-m_0}-1)|C|, & {\rm if}~i= M+1.
                       \end{array}
\right.$}
\end{center}

  \item [{\rm(3)}] $C$ is regular if and only if $C_0$ is regular.

  \item [{\rm(4)}] $C$ is projective if and only if $C_0$ is projective.

  \item [{\rm(5)}] $C$ is Plotkin-optimal if and only if $C_0$ is Plotkin-optimal for any $m_0$.

\end{itemize}
\end{theorem}

\begin{proof}
Suppose that $G=[{\bf g}_1~|~{\bf g}_2~|~\cdots~|~{\bf g}_n]$ is an $\ell\times n$ matrix, where $\ell=k_1+k_2+\cdots+k_m$. Then
${\bf g}_{i,a}=\left[\begin{array}{c}
         {\bf g}_i \\
         a
       \end{array}
\right]$ is any column of $G_0$, where $1\leq i\leq n$ and $a\in \langle\theta^{m_0}\rangle$.

(1) and (2). Applying Lemma \ref{A} to the code $C$, we can obtain the desired result. Here we omit the proof.

(3). Suppose that $C$ is regular. Then we have $\{{\bf x}\cdot {\bf g}_i~|~{\bf x}\in R^\ell\}=R$, where $1\leq i\leq n$. It turns out that
$$R=\{{\bf x}\cdot {\bf g}_i~|~{\bf x}\in R^\ell\}\subseteq \{{\bf x}\cdot {\bf g}_{i,a}~|~{\bf x}\in R^{\ell+1}\}\subseteq R,$$
 where $1\leq i\leq n$ and $a\in \langle\theta^{m_0}\rangle$. Hence $\{{\bf x}\cdot {\bf g}_{i,a}~|~{\bf x}\in R^{\ell+1}\}=R$ and $C_0$ is regular.

Conversely, suppose that $C_0$ is regular.
Then $\{{\bf x}\cdot {\bf g}_{i,a}~|~{\bf x}\in R^{\ell+1}\}=R$ for any $1\leq i\leq n$ and $a\in \langle\theta^{m_0}\rangle$.
When $a=0\in \langle\theta^{m_0}\rangle$,
$$\{{\bf x}\cdot {\bf g}_i~|~{\bf x}\in R^\ell\}=\{{\bf x}\cdot {\bf g}_{i,a}~|~{\bf x}\in R^{\ell+1}\}=R,$$
where $1\leq i\leq n$. Hence $C$ is regular.

(4). Suppose that $C$ is projective. Then we have ${\bf g}_iR\neq {\bf g}_jR$ for any pair of distinct coordinates $i,j\in\{1,2,\ldots,n\}$. It follows that ${\bf g}_{i,a}R\neq {\bf g}_{j,b}R$ for any pair of distinct coordinates $i,j\in\{1,2,\ldots,n\}$ and $a,b\in \langle\theta^{m_0}\rangle$. In addition, it is easy to see that ${\bf g}_{i,a}R\neq {\bf g}_{i,b}R$ for any pair of distinct elements $a,b\in \langle\theta^{m_0}\rangle$ and $1\leq i\leq n$. Hence $C_0$ is projective.

Conversely, suppose that $C_0$ is projective.
Then we have ${\bf g}_{i,0}R\neq {\bf g}_{j,0}R$ for any pair of distinct coordinates $i,j\in\{1,2,\ldots,n\}$.
Hence ${\bf g}_{i}R\neq {\bf g}_{j}R$ for any pair of distinct coordinates $i,j\in\{1,2,\ldots,n\}$. Therefore, $C$ is projective.

(5). It is easy to see that
$${\rm d_{ hom}}(C_0)=q^{m-m_0}{\rm d_{ hom}}(C).$$
Thus ${\rm d_{ hom}}(C_0)=(q^{m-1}-q^{m-2})q^{m-m_0}n$ if and only if ${\rm d_{ hom}}(C)=(q^{m-1}-q^{m-2})n$. That is to say, $C$ is Plotkin-optimal if and only if $C_0$ is Plotkin-optimal for any $m_0$.
\end{proof}


\begin{cor}\label{cor-optimal}
Let $C$ be a Plotkin-optimal regular projective linear $[n;k_0,k_1,\ldots,k_{m-1}]_q$ code over $R$ with the homogeneous weight distribution
$$1+A_{\omega_1}X^{\omega_1}+A_{\omega_2}X^{\omega_2}+\cdots+A_{\omega_M}X^{\omega_M},$$
where $\omega_1=(q^{m-1}-q^{m-2})n.$
Let $s_i$ be a nonnegative integer for $0\leq i\leq m-1$. Then there exists a
Plotkin-optimal regular projective linear $[q^sn;k_0+s_0,k_1+s_1,\ldots,k_{m-1}+s_{m-1}]_q$ code $C'$ over $R$ with the homogeneous weight distribution
$$1+B_{\omega'_1}X^{\omega'_1}+B_{\omega'_2}X^{\omega'_2}+\cdots+
      B_{\omega'_M}X^{\omega'_M},$$
where $s=ms_0+(m-1)s_1+\cdots+s_{m-1}$, $\omega'_i=q^s\omega_i$ and $$B_{w'_i}=\left\{\begin{array}{ll}
                         A_{w_i}, & {\rm if}~i\neq 1, \\
                         |C'|-1\sum_{i=2}^{M}A_{\omega_M}, & {\rm if}~i= 1.
                       \end{array}
\right.$$
\end{cor}

\begin{proof}
The proof is straightforward from Theorem \ref{thm-1}, so we omit it.
\end{proof}

\section{Plotkin-optimal two-homogeneous weight linear codes over $R$}

\subsection{Possible parameters of Plotkin-optimal two-homogeneous weight regular projective linear codes over $R$}
Suppose that $m\geq 2$. Let $R$ be a finite chain ring of the depth $m$ with the residue field $\F_q$ and the maximum idea $\langle \theta\rangle$. The {\em distance matrix} of a code $C$ over $R$ is the $|C|\times |C|$ matrix $D$ with rows and columns indexed by the elements of $C$ and whose $(u,v)$-th entry is $D_{uv}={\rm wt_{hom}}(u-v)$. The following theorem characterizes some properties of the distance matrix of a code.

\begin{theorem}{\rm \cite[Theorem 13]{FFA-2-weight}}\label{thm-2-equ}
Let $C$ be a regular projective linear code over $R$. Let $J$ denote the all-ones matrix of order $|C|$. Then
\begin{itemize}
  \item [{\rm(1)}] $DJ=n(q^{m-1}-q^{m-2})|C|J$ and
  \item [{\rm(2)}] $D^2+\frac{(q^{m-1}-q^{m-2})|C|}{|R^\times|}D=
      n(q^{m-1}-q^{m-2})^2|C|\left(\frac{1}{|R^\times|}+n\right)J$.
\end{itemize}
\end{theorem}

\begin{prop}\label{prop-equ}
The nonzero weights $\omega_1,\omega_2$ of a two-homogeneous weight linear code $C$ of length $n$ over $R$ satisfy the relation
$$(\omega_1+\omega_2)n(q^{m-1}-q^{m-2})|C|=(q^{m-1}-q^{m-2})^2 \left(\frac{n}{q^m-q^{m-1}}+n^2\right)|C|+\omega_1\omega_2(|C|-1).$$
\end{prop}

\begin{proof}
It is easy to see that
$$A_{\omega_1}+A_{\omega_2}=|C|-A_0=|C|-1.$$
By Theorem \ref{thm-2-equ}, we have
\begin{align*}
 A_{\omega_1}\omega_1+A_{\omega_2}\omega_2&=n(q^{m-1}-q^{m-2})|C|,\\
 A_{\omega_1}\omega_1^2+A_{\omega_2}\omega_2^2&=(q^{m-1}-q^{m-2})^2 \left(\frac{n}{q^m-q^{m-1}}+n^2\right)|C|.
\end{align*}
By the equation $(\omega_1+\omega_2)(A_{\omega_1}\omega_1+A_{\omega_2}\omega_2)=
(A_{\omega_1}\omega_1^2+A_{\omega_2}\omega_2^2)+\omega_1\omega_2(A_{\omega_1}+A_{\omega_2})$, we can obtain the desired result.
\end{proof}

\begin{cor}\label{cor-two-weight}
If $C$ is a Plotkin-optimal regular projective two-homogeneous weight code over $R$ with the nonzero weights $\omega_1<\omega_2$, then $\omega_1=(q^{m-1}-q^{m-2})n$ and $\omega_2=\frac{|C|}{q}.$
\end{cor}
\begin{proof}
By Lemma \ref{lemma-P-optimal} and Proposition \ref{prop-equ}, the result follows.
\end{proof}

In order to characterize possible parameters of Plotkin-optimal two-homogeneous weight regular projective linear codes over $R$, we introduce an important lemma.

\begin{lemma}\label{lem-111}
Let $k$ be a positive integer and let ${\bf S_{u}}={\bf u}+\langle\theta^{m-1}\rangle^{k}$ for some ${\bf u}\in R^{k}$. If ${\bf c}\in R^{k}\backslash \langle\theta\rangle^{k}$, then
$$\sum_{{\bf x}\in {\bf S_{u}}}{\rm wt_{hom}}({\bf x}\cdot{\bf c})=
  q^{{k}}(q^{m-1}-q^{m-2}).$$
\end{lemma}

\begin{proof}
{\bf Step 1.} We consider the case for ${\bf u}\in \langle\theta^{m-1}\rangle^{k}$. Then
$${\bf S_{u}}={\bf u}+\langle\theta^{m-1}\rangle^{k}=\langle\theta^{m-1}\rangle^{k}={\bf S_{0}}.$$
Let $G$ be the ${k}\times (q^{{k}}-1)$ matrix over $\langle\theta^{m-1}\rangle^{k}$ with all distinct nonzero columns. Hence $\mu_G({\bf x})=0$ for ${\bf x}\in R^{k}\backslash \langle\theta^{m-1}\rangle^{k}.$ By Lemma \ref{lem-one-weight}, $G$ generates a one-homogeneous weight linear $[q^{{k}}-1;0,0,\ldots,0,k]_q$ code $C$ over $R$ with the nonzero homogeneous weight $q^{k}(q^{m-1}-q^{m-2})$. Note that ${\bf c}G$ is a nonzero codeword of $C$. Hence
$$q^{k}(q^{m-1}-q^{m-2})={\rm wt_{hom}}({\bf c}G)=\sum_{{\bf x}\in R^k}\mu_{G}({\bf x}){\rm wt_{hom}}({\bf x}\cdot{\bf c})=\sum_{{\bf x}\in {\bf S_{0}}}{\rm wt_{hom}}({\bf x}\cdot{\bf c}).$$

{\bf Step 2.} Suppose that ${\bf u}\notin \langle\theta^{m-1}\rangle^k.$ This implies that ${\bf 0}\notin {\bf S_u}$. Note that $\theta{\bf x}=\theta{\bf y}=\theta{\bf u}$ for any ${\bf x},{\bf y}\in {\bf S_u}$. We consider two cases.

Case 1. If there exists ${\bf x}\in {\bf S_u}$ such that ${\rm wt_{hom}}({\bf x}\cdot {\bf c})=q^{m-1}-q^{m-2}$, then for any ${\bf y}\in {\bf S_u}$ we have $\theta{\bf y}\cdot {\bf c}=\theta{\bf x}\cdot {\bf c}\neq 0$, $i.e.$, ${\bf y}\cdot {\bf c}\notin \langle\theta^{m-1}\rangle$. So ${\rm wt_{hom}}({\bf y}\cdot {\bf c})=q^{m-1}-q^{m-2}$ for any ${\bf y}\in {\bf S_u}$. Hence we have
$$\sum_{{\bf x}\in {\bf S_{u}}}{\rm wt_{hom}}({\bf x}\cdot{\bf c})=(q^{m-1}-q^{m-2})|{\bf S_u}|=q^k(q^{m-1}-q^{m-2}).$$

Case 2. If there exists ${\bf x}\in {\bf S_u}$ such that ${\rm wt_{hom}}({\bf x}\cdot {\bf c})\neq q^{m-1}-q^{m-2}$, $i.e.$, ${\rm wt_{hom}}({\bf x}\cdot {\bf c})\in\{0, q^{m-1}\}$, then we have ${\rm wt_{hom}}({\bf y}\cdot {\bf c})\in\{0, q^{m-1}\}$ for any ${\bf y}\in {\bf S_u}$ by Case 1. This implies that ${\bf y}\cdot {\bf c}\in \langle\theta^{m-1}\rangle$ for any ${\bf y}\in {\bf S_u}$ by the definition of the homogeneous weight.
Since ${\bf c}\in R^k\backslash \langle\theta\rangle^k$, there exists $1\leq i_0\leq k$ such that $c_{i_0}\in R\backslash \langle\theta\rangle$, where $c_{i_0}$ is the $i_0$-th coordinate component of ${\bf c}$. Let ${\bf e}_{i_0}=(e_{i_01},e_{i_02},\ldots,e_{i_0k})$, where $e_{i_0j}=\delta_{i_0j}$ for $1\leq j\leq k$.

Let $A:=\{{\bf x}\in {\bf S_u}~|~{\rm wt_{hom}}({\bf x}\cdot {\bf c})=0\}$ and $B:=\{{\bf x}\in {\bf S_u}~|~{\rm wt_{hom}}({\bf x}\cdot {\bf c})=q^{m-1}\}$.
For any fixed ${\bf x}\in {\bf S_u}$, it can be checked that ${\bf x}+a{\bf e}_{i_0}\in {\bf S_u}$ for any $a\in \langle\theta^{m-1}\rangle$ and
$$\{({\bf x}+a{\bf e}_{i_0})\cdot {\bf c}~|~a\in \langle\theta^{m-1}\rangle\}=
\{{\bf x}\cdot {\bf c}+a~|~a\in \langle\theta^{m-1}\rangle\}=
\{a~|~a\in \langle\theta^{m-1}\rangle\}.$$
Thus, there is exactly one element of weight 0 in set $\{({\bf x}+a{\bf e}_{i_0})\cdot {\bf c}~|~a\in \langle\theta^{m-1}\rangle\}$.
It follows that $|B|=(q-1)|A|$. Hence ${\bf S_u}=A\cup B$ and $A=\frac{|{\bf S_u}|}{q}=q^{k-1}$ and $B=q^k-q^{k-1}$. It turns out that
\begin{align*}
 \sum_{{\bf x}\in {\bf S_{u}}}{\rm wt_{hom}}({\bf x}\cdot{\bf c})& =\sum_{{\bf x}\in A}{\rm wt_{hom}}({\bf x}\cdot{\bf c})+\sum_{{\bf x}\in B}{\rm wt_{hom}}({\bf x}\cdot{\bf c}) \\
   & =q^{m-1}|B|\\
   &=q^{m-1}(q^k-q^{k-1})\\
   &=q^k(q^{m-1}-q^{m-2}).
\end{align*}
This completes the proof.
\end{proof}

\begin{cor}\label{cor-111}
Let $C$ be a regular linear $[n;k_0,k_1,\ldots,k_{m-1}]_q$ code over $R$ with a generator matrix $G$. Let $k'=k_0+k_1+\cdots+k_{m-1}$ and ${\bf S_{u}}={\bf u}+\langle\theta^{m-1}\rangle^{k'}$ for some ${\bf u}\in R^{k'}$. Then
$$\sum_{{\bf x}\in {\bf S_u}}{\rm wt_{hom}}({\bf x}G)=
  q^{k'}(q^{m-1}-q^{m-2})n.
$$
\end{cor}

\begin{proof}
Since $C$ is regular, $\mu_G({\bf c})=0$ for ${\bf c}\in \langle\theta\rangle^{k'}$. By Lemma \ref{lem-111}, we have
\begin{align*}
  \sum_{{\bf x}\in {\bf S_u}}{\rm wt_{hom}}({\bf x}G) & =
\sum_{{\bf x}\in {\bf S_u}}\sum_{{\bf c}\in R^{k'}}\mu_{G}({\bf c}){\rm wt_{hom}}({\bf c}\cdot{\bf x}) \\
   & =\sum_{{\bf c}\in R^{k'}}\mu_{G}({\bf c})\left(\sum_{{\bf x}\in {\bf S_u}}{\rm wt_{hom}}({\bf c}\cdot{\bf x})\right)\\
   &= q^{k'}(q^{m-1}-q^{m-2})n.
\end{align*}
This completes the proof.
\end{proof}

\begin{cor}\label{cor-222}
Let $C$ be a regular projective linear $[n;k_0,k_1,\ldots,k_{m-1}]_q$ code over $R$ with a generator matrix $G$. Let $k'=k_0+k_1+\cdots+k_{m-1}$. If $C$ is Plotkin-optimal, then
${\rm wt_{hom}}({\bf x}G)=(q^{m-1}-q^{m-2})n$ for all ${\bf x}\in R^{k'}\backslash\langle\theta^{m-1}\rangle^{k'}$.
\end{cor}

\begin{proof}
Assume that ${\bf S_u}={\bf u}+ \langle\theta^{m-1}\rangle^{k'}$, where ${\bf u}\notin \langle\theta^{m-1}\rangle^{k'}$. Since $C$ is Plotkin-optimal, ${\rm d_{hom}}(C)=(q^{m-1}-q^{m-2})n$ by Lemma \ref{lemma-P-optimal}.
By Corollary \ref{cor-111}, we have
$$q^{k'}(q^{m-1}-q^{m-2})n=\sum_{{\bf x}\in {\bf S_u}}{\rm wt_{hom}}({\bf x}G)\geq |{\bf S_u}|{\rm d_{hom}}(C)=q^{k'}(q^{m-1}-q^{m-2})n.
.$$
Hence ${\rm wt_{hom}}({\bf x}G)=(q^{m-1}-q^{m-2})n$ for any ${\bf x}\in {\bf S_u}$. Since ${\bf u}$ is arbitrary and ${\bf u}\notin \langle\theta^{m-1}\rangle^{k'}$, ${\rm wt_{hom}}({\bf x}G)=(q^{m-1}-q^{m-2})n$ for all ${\bf x}\in R^{k'}\backslash\langle\theta^{m-1}\rangle^{k'}$.
\end{proof}

The following theorem shows possible parameters of Plotkin-optimal two-homogeneous weight regular projective codes over a finite chain ring.

\begin{theorem}
Let $C$ be a Plotkin-optimal two-homogeneous weight regular projective linear $[n;k_0,k_1,\ldots,k_{m-1}]_q$ code over $R$ with a generator matrix $G$. Then $n=\frac{q^{k-t}(q^t-1)}{q^m-q^{m-1}}$ for some integer $1\leq t\leq k_0$, where $k=mk_0+(m-1)k_1+\cdots+k_{m-1}$ is the $q$-dimension of $C$. Moreover, the homogeneous weight distribution is as follows
$$
\begin{array}{ll}
\hline
{\rm Homogeneous\ weight}\ \omega &  {\rm Multiplicity}  \\ \hline
0 & 1 \\
q^{k-t-1}(q^t-1) & q^k-q^t  \\
q^{k-1} & q^t-1\\ \hline
\end{array}
$$
\end{theorem}

\begin{proof}
Let $T$ be the Teichm${\rm \ddot{ u}}$ller set of $R$. Then $0\in T$ and $|T|=q$.
Let $\omega_1<\omega_2$ be the two nonzero weights of $C$.
By Corollary \ref{cor-two-weight}, we have $\omega_1=(q^{m-1}-q^{m-2})n$ and $\omega_2=\frac{|C|}{q}$. By Corollarys \ref{cor-111} and \ref{cor-222}, we have
$$q^{k'}(q^{m-1}-q^{m-2})n=\sum_{{\bf x}\in {\bf S_0}}{\rm wt_{hom}}({\bf x}G)=A_{\omega_2}\omega_2+(q^{k'}-1-A_{\omega_2})\omega_1.
$$
From the above equation, we have $n=\frac{A_{\omega_2}|C|}{(A_{\omega_2}+1)(q^m-q^{m-1})}$. This implies that $A_{\omega_2}+1$ divides $A_{\omega_2}|C|$, and it turns out that $A_{\omega_2}+1$ divides $|C|$. For any ${\bf c}\in C$ and $a\in T\backslash \{0\}$, we have $a{\bf c}\in C$ and ${\rm wt_{hom}} (a{\bf c})={\rm wt_{hom}} ({\bf c})$. Hence $A_{\omega_2}$ is divisible by $q-1$. We conclude that $A_{\omega_2}+1=q^t$ for some integer $t$. Hence $t\geq 1$, otherwise $A_{\omega_2}=0$, which is a contradiction. Let $k=mk_0+(m-1)k_1+\cdots+k_{m-1}$ be the $q$-dimension of $C$. By Lemma \ref{lemma-n}, we have
$$\frac{q^{k}-q^{k-k_0}}{q^m-q^{m-1}}\geq n=\frac{(q^t-1)q^k}{q^t(q^m-q^{m-1})},$$
which implies that $1\leq t\leq k_0$. Therefore, $n=\frac{q^{k-t}(q^t-1)}{q^m-q^{m-1}}$, $\omega_1=(q^{m-1}-q^{m-2})n=q^{k-t-1}(q^t-1)$, $\omega_2=\frac{|C|}{q}=q^{k-1}$, $A_{\omega_1}=|C|-A_{\omega_2}-1=q^k-q^t$ and $A_{\omega_2}=q^t-1$.
\end{proof}

\subsection{The construction of Plotkin-optimal two-homogeneous weight regular projective linear codes over $R$}
In this section, we will give the construction of two-homogeneous weight codes and study some properties of these codes.
Let $R=\{0,r_1,r_2,\ldots,r_{p^m-1}\}$ be a finite chain ring of the depth $m$ with the residue field $\F_q$ and the maximum idea $ \langle \theta\rangle$, where $m\geq 2$.
Let $k$ be a positive integer.
We define the matrix $Y_k$ by inductive constructions as follows:
\begin{center}
$
 Y_1=[1]$ and $Y_k=\begin{bmatrix}
    \begin{array}{ccccc|c}
    Y_{k-1} & Y_{k-1} &Y_{k-1} & \cdots & Y_{k-1} &B_{k-1} \\
    \textbf 0 & r_1\textbf 1& r_2\textbf 1 & \cdots & r_{q^{m}-1}\textbf 1 & \textbf 1'
    \end{array}
    \end{bmatrix},
    $
    \end{center}
where $B_{k-1}$ is a $(k-1)\times q^{(m-1)(k-1)}$ matrix over $\langle \theta\rangle$ consisting of all different columns, $\textbf 0$ is the all-zero vector of length $\frac{q^{m(k-1)}-q^{(m-1)(k-1)}}{q^m-q^{m-1}}$, $\textbf 1$ is the all-one vector of length $\frac{q^{m(k-1)}-q^{(m-1)(k-1)}}{q^m-q^{m-1}}$, $\textbf 1'$ is the all-one row vector of length $q^{(m-1)(k-1)}$.
It can be checked that $Y_k$ is a $k\times \frac{q^{mk}-q^{(m-1)k}}{q^m-q^{m-1}}$ matrix over $R$.

\begin{prop}\label{y_k}
Suppose that $k$ is a positive integer. Let $C$ be a linear code with the generator matrix $Y_k$. Then $C$ is a free Plotkin-optimal two-homogeneous weight code of length $n=\frac{q^{mk}-q^{(m-1)k}}{q^m-q^{m-1}}$. Moreover, the homogeneous weight distribution is as follows.
$$
\begin{array}{ll}
\hline
{\rm Homogeneous\ weight}\ \omega &  {\rm Multiplicity}  \\ \hline
0 & 1 \\
q^{mk-1}-q^{(m-1)k-1} & q^{mk}-q^k  \\
q^{mk-1} & q^k-1\\ \hline
\end{array}
$$
\end{prop}

\begin{proof}
Consider the matrix $S_k$ with the block form
$$S_k=\left[[\alpha Y_k]_{\alpha\in R^\times}|~ B_k\right].$$
Then $S_k$ is a $k\times q^{mk}$ matrix over $R$. By Theorem \ref{Y_k projective}, we know that the code generated by $Y_k$ is regular and projective, that is,
$$\mu_{Y_k}({\bf x})=0~{\rm for}~{\bf x}\in \langle \theta\rangle^k~{\rm and}~
\sum_{\alpha\in R^\times}\mu_{Y_k}(\alpha {\bf x})\leq 1~ {\rm for}~{\bf x}\in R^k\backslash \langle \theta\rangle^k.$$
It can be checked that
$$\mu_{S_k}({\bf x})=1~{\rm for}~{\bf x}\in \langle \theta\rangle^k~{\rm and}~
\mu_{S_k}({\bf x})\leq 1~{\rm for}~{\bf x}\in R^k\backslash\langle \theta\rangle^k.$$
Hence
$$q^{mk}=\sum_{{\bf x}\in R^k}\mu_{S_k}({\bf x})=\sum_{{\bf x}\in \langle \theta\rangle^k}\mu_{S_k}({\bf x})+\sum_{{\bf x}\in R^k\backslash \langle \theta\rangle^k}\mu_{S_k}({\bf x})\leq q^{(m-1)k}+(q^{mk}-q^{(m-1)k})\leq q^{mk}.$$
It turns out that $\mu_{S_k}({\bf x})=1~{\rm for}~{\bf x}\in R^k$, that is, all columns of $S_k$ are pairwise different. If we delete the zero column from $S_k$ (in fact we delete the zero column from $B_k$), then we get a $k\times (q^{mk}-1)$ matrix $S'_k$ over $R$. Therefore, $S'_k$ is a matrix whose columns are all nonzero vectors of $R^k$.

Let $\mathcal{S}'_k$ be a code with the generator matrix $S'_k$. It is easy to see that $\mathcal{S}'_k$ is a special case of the code in Lemma \ref{lem-one-weight}, when $k_1=k, k_2=k_3=\cdots=k_m=0$. Then $\mathcal{S}'_k$ is a one-homogeneous weight linear code of length $q^{mk}-1$ with the nonzero homogeneous weight $\omega'=q^{mk}(q^{m-1}-q^{m-2})$.

Let $B_k'$ denote the matrix which is obtained by deleting the zero column from $B_k$. Let $\mathcal{B}'_k$ be a code with the generator matrix $B'_k$. It is easy to check that $\mathcal{B}'_k$ is a special case of the code in Lemma \ref{lem-one-weight}, when $k_1=0, k_2=k,k_3=\cdots=k_{m}=0$. Then $\mathcal{B}'_k$ is a one-homogeneous weight linear code of length $q^{(m-1)k}-1$ with the nonzero homogeneous weight $\omega''=q^{(m-1)k}(q^{m-1}-q^{m-2})$.

Let $\textbf s=((\alpha\textbf c)_{\alpha\in R^\times}, \textbf b)\in \mathcal{S}'_k$ be a nonzero codeword, where $\textbf c \in C$ and $\textbf b \in \mathcal{B}'_k$. Then we obtain
$${\rm wt_{hom}}(\textbf s)=\sum_{\alpha\in R^\times}{\rm wt_{hom}}(\alpha \textbf c)+{\rm wt_{hom}}(\textbf b)=(q^m-q^{m-1}){\rm wt_{hom}}(\textbf c)+{\rm wt_{hom}}(\textbf b).$$
Moreover, we have
$${\rm wt_{hom}}(\textbf c)=\frac{{\rm wt_{hom}}(\textbf s)-{\rm wt_{hom}}(\textbf b)} {q^m-q^{m-1}}=\frac{q^{mk}(q^{m-1}-q^{m-2})-{\rm wt_{hom}}(\textbf b)}{q^m-q^{m-1}}.$$
Since the code $\mathcal{B}'_k$ is a one-homogeneous weight code with the nonzero homogeneous weight $\omega''=q^{(m-1)k}(q^{m-1}-q^{m-2})$, ${\rm wt_{hom}}(\textbf b)=q^{(m-1)k}(q^{m-1}-q^{m-2})$ or $0$. It follows that
\begin{itemize}
  \item If ${\rm wt_{hom}}(\textbf b)=q^{(m-1)k}(q^{m-1}-q^{m-2})$, then ${\rm wt_{hom}}(\textbf c)=\omega_1=q^{mk-1}-q^{(m-1)k-1}$.
  \item If ${\rm wt_{hom}}(\textbf b)=0$, then ${\rm wt_{hom}}(\textbf c)=\omega_2=q^{mk-1}$. It is easy to get $A_{\omega_2}=q^k-1$ when ${\rm wt_{hom}}(\textbf b)=0$ and ${\rm wt_{hom}}({\bf s})\neq 0$.
\end{itemize}
 Hence $A_{\omega_1}=q^{mk}-A_{\omega_2}-1=q^{mk}-q^k$. By Lemma \ref{lemma-P-optimal}, we have that $C$ is Plotkin-optimal. This completes the proof.
\end{proof}

\begin{cor}
Suppose that $k$ is a positive integer. Let $C$ be a linear code over $R$ with the generator matrix $Y_k$. For any $\textbf c \in C$ and
$\textbf c \neq {\rm {\bf0}}$, then we have
\begin{align*}
\begin{split}
{\rm wt_{hom}}({\bf x}Y_k)=\left\{
\begin{array}{ll}
q^{mk-1}-q^{(m-1)k-1}, & if~{\bf x}\in R^k\backslash\langle\theta^{m-1}\rangle, \\
q^{mk-1}, & if ~{\bf x}\in \langle\theta^{m-1}\rangle\backslash\{{}\bf 0\},\\
0,&if~{\bf x}={\bf 0}.
\end{array}
\right.
\end{split}
\end{align*}
\end{cor}

\begin{proof}
If ${\bf x}={\bf 0}$, then the result follows.
By Proposition \ref{y_k}, we have ${\rm wt_{hom}}({\bf x}Y_k)=q^{mk-1}$ if and only if ${\bf x}\in \langle\theta^{m-1}\rangle\backslash\{{\bf 0}\}$. Hence ${\rm wt_{hom}}({\bf x}Y_k)=q^{mk-1}-q^{(m-1)k-1}$ if ${\bf x}\in R^k\backslash\langle\theta^{m-1}\rangle.$
\end{proof}

\begin{example}
Let $R=\Z_{2^3}$. By Proposition \ref{y_k}, there is a free Plotkin-optimal two-homogeneous weight linear $[12;2,0,0]$ code over $\Z_{2^3}$ with the following generator matrix
$$\begin{pmatrix}
    \begin{array}{ccccccccccc}
    1 & 1 & 1 & 1 & 1 & 1 & 1 & 1 & 0246  \\
    0 & 1 & 2 & 3 & 4 & 5 & 6 & 7 & 1111
    \end{array}
    \end{pmatrix}.$$
Moreover, $\omega_1=24$, $\omega_2=32$, $A_{24}=60$ and $A_{32}=3$.
\end{example}

\begin{example}
Let $R=\F_2+u\F_2=\{0,1,u,\overline{u}\}$, where $u^2=0$ and $\overline{u}=u+1$. By Theorem \ref{thm two Lee weight Y_k over Z_2^k}, there is a free Plotkin-optimal two-homogeneous weight linear $[6;2,0]$ code over $R$ with the following generator matrix
$$G=\begin{pmatrix}
    \begin{array}{ccccc}
    1 & 1 & 1 & 1 &  0u\\
    0 & 1 & u & \overline{u} & 11
    \end{array}
    \end{pmatrix}.$$
Moreover, $\omega_1=6$, $\omega_2=8$, $A_{6}=12$ and $A_{8}=3$.
\end{example}

\begin{theorem}\label{thm two Lee weight Y_k over Z_2^k}
Suppose that $k_i$ is a nonnegative integer, where $0\leq i\leq m-1$. Then for any $1\leq t\leq k_0$, there exists a Plotkin-optimal two-homogeneous weight regular projective linear $[n;k_0,k_1,\ldots,k_{m-1}]_q$ code over $R$, where $n=\frac{q^{k-t}(q^t-1)}{q^m-q^{m-1}}$ and $k=mk_0+(m-1)k_1+\cdots+k_{m-1}$. Moreover, the homogeneous weight distribution is as follows.
$$
\begin{array}{ll}
\hline
{\rm Homogeneous\ Weight}\ \omega &  {\rm Multiplicity}  \\ \hline
0 & 1 \\
q^{k-t-1}(q^t-1) & q^k-q^t  \\
q^{k-1} & q^t-1\\ \hline
\end{array}
$$
\end{theorem}

\begin{proof}
The result follows by applying Corollary \ref{cor-optimal} with $s_0=k_0-t$, $s_i=k_i$ ($1\leq i\leq m-1$) on the linear code generated by $Y_t$.
\end{proof}

\begin{example}
Let $R=\Z_{2^3}$ and $t=2$. By Theorem \ref{thm two Lee weight Y_k over Z_2^k}, there is a Plotkin-optimal two-homogeneous weight linear $[192;3,0,1]$ code over $\Z_{2^3}$ with the following generator matrix
$$G=\begin{pmatrix}
    \begin{array}{cccccccc|cccccccc}
    Y_2 & Y_2 & Y_2 & Y_2 & Y_2 & Y_2 & Y_2 & Y_2
    &Y_2 & Y_2 & Y_2 & Y_2 & Y_2 & Y_2 & Y_2 & Y_2 \\
     \textbf 0 & \textbf 1 & \textbf 2 & \textbf 3 & \textbf 4 & \textbf 5 & \textbf 6 & \textbf 7&
     \textbf 0 & \textbf 1 & \textbf 2 & \textbf 3 & \textbf 4 & \textbf 5 & \textbf 6 & \textbf 7\\
     \hline
     \textbf 0 & \textbf 0 & \textbf 0 & \textbf 0 & \textbf 0 & \textbf 0 & \textbf 0 & \textbf 0&
     \textbf 4 & \textbf 4 & \textbf 4 & \textbf 4 & \textbf 4 & \textbf 4 & \textbf 4 & \textbf 4
    \end{array}
    \end{pmatrix},$$
where $\textbf i$ is the row vector $(i,i,\cdots,i)$ of length 12 for any $i\in \Z_{2^3}$ and the matrix $Y_2$ is of the form
$$\begin{pmatrix}
    \begin{array}{ccccccccccc}
    1 & 1 & 1 & 1 & 1 & 1 & 1 & 1 & 0246  \\
    0 & 1 & 2 & 3 & 4 & 5 & 6 & 7 & 1111
    \end{array}
    \end{pmatrix}.$$
Moreover, $\omega_1=384$, $\omega_2=512$, $A_{384}=1020$ and $A_{512}=3$.
\end{example}

\begin{example}
Let $R=\Z_{3^2}$ and $t=2$. By Theorem \ref{thm two Lee weight Y_k over Z_2^k}, there is a Plotkin-optimal two-homogeneous weight linear $[36;2,1]$ code over $\Z_{3^2}$ with the following generator matrix
$$G=\begin{pmatrix}
    \begin{array}{ccccccccc}
    Y_2 & Y_2 & Y_2 \\
     \textbf 0 & \textbf 3 & \textbf 6
    \end{array}
    \end{pmatrix},$$
where $\textbf i$ is the row vector $(i,i,\cdots,i)$ of length 12 for any $i\in \Z_{3^2}$ and the matrix $Y_2$ is of the form
$$\begin{pmatrix}
    \begin{array}{ccccccccccc}
    1 & 1 & 1 & 1 & 1 & 1 & 1 & 1 & 1 & 036  \\
    0 & 1 & 2 & 3 & 4 & 5 & 6 & 7 & 8 & 111
    \end{array}
    \end{pmatrix}.$$
Moreover, $\omega_1=72$, $\omega_2=81$, $A_{72}=234$ and $A_{81}=8$.
\end{example}

\subsection{Regular projective codes and Gray images}

In this subsection, we study that the regularity, projectivity and Gray images of the codes in Theorem \ref{thm two Lee weight Y_k over Z_2^k}.

\begin{theorem}\label{Y_k projective}
Let $C$ be a linear $[n;k_0,k_1,\ldots,k_{m-1}]$ code over $R$ obtained in Theorem \ref{thm two Lee weight Y_k over Z_2^k}. Then $C$ is regular and projective.
\end{theorem}

\begin{proof}
By (3) and (4) of Theorem \ref{thm-1}, it suffices to prove that the code $C_0$ generated by $Y_k$ is regular and projective, i.e., the code in Proposition \ref{y_k} is regular and projective.
Suppose that $Y_k=[{\bf g}_{k,1}~|~{\bf g}_{k,2}~|~\cdots~|~{\bf g}_{k,n_k}]$.

(i) By the inductive construction of $Y_k$, we know that every column of $Y_k$ contains at least one unit element of $R$. Hence $\{{\bf x}\cdot {\bf g}_{k,i}~|~{\bf x}\in R^k\}=R$ for $1\leq i\leq n_k$. It turns out that $C$ is regular.

(ii) We prove the code generated by $Y_{k}$ is projective by induction on $k$. If $k=1$, then the result is clearly true. Assume that the result is valid for $k=t$.
Next, we prove that the result is valid for $k=t+1$. For $k=t+1$, we know that
$$
    Y_{t+1}=\begin{bmatrix}
    \begin{array}{cccccccc}
    Y_{t} & Y_{t} & Y_{t} & \cdots & Y_{t} &B_{t} \\
    \textbf 0 & r_1\textbf 1& r_2\textbf 1 & \cdots & r_{q^{m}-1}\textbf 1 & \textbf 1'
    \end{array}
    \end{bmatrix},
    $$
where $R=\{0,r_1,r_2,\ldots,r_{p^m-1}\}$ and $B_{t}$ is a $t\times q^{(m-1)t}$ matrix over $\langle\theta\rangle$ consisting of all different columns. Note that $\textbf 0$ is the all-zero vector of length $n_t=\frac{q^{mt}-q^{(m-1)t}}{q^m-q^{m-1}}$, $\textbf 1$ is the all-one vector of length $n_t=\frac{q^{mt}-q^{(m-1)t}}{q^m-q^{m-1}}$ and $\textbf 1'$ is the all-one vector of length $q^{(m-1)t}$.
Note that
\begin{center}
${\bf g}_{t+1,an_t+i}=\left[
\begin{array}{c}
  {\bf g}_{t,i} \\
  r_a
\end{array}
\right]$ and ${\bf g}_{t+1,q^mn_t+j}=\left[
\begin{array}{c}
  {\bf b}_j \\
  1
\end{array}
\right],$
\end{center}
where $1\leq i\leq n_t$, $0\leq a\in q^m-1$ and ${\bf b}_j$ is the $j$-th column of $B_t$ for $1\leq j\leq q^{(m-1)t}$. We let
\begin{center}
$\mathcal{H}_1=\{{\bf g}_{t+1,j}~|~1\leq j\leq q^mn_t\}$ and
$\mathcal{H}_2=\{{\bf g}_{t+1,j}~|~q^mn_t+1\leq j\leq n_{t+1}\}$.
\end{center}

It suffices to prove that ${\bf h}_1R\neq {\bf h}_2R$ for any two different columns of $Y_{t+1}$ if we want to prove that the code generated by $Y_{t+1}$ is projective. So we consider the following three cases.

\begin{itemize}
  \item Case 1. Suppose that ${\bf h}_1,{\bf h}_2\in \mathcal{H}_1$.
  By (4) of Theorem \ref{thm-1}, the code generated by the first $q^mn_t$ columns of $Y_{t+1}$ is projective. Hence ${\bf h}_1R\neq {\bf h}_2R$.
  \item Case 2. Suppose that ${\bf h}_1,{\bf h}_2\in \mathcal{H}_2$. Without loss of generality, assume that
       \begin{center}
       ${\bf h}_1={\bf g}_{t+1,q^mn_t+i_1}=\left[
\begin{array}{c}
  {\bf b}_{i_1} \\
  1
\end{array}
\right]$ and ${\bf h}_2={\bf g}_{t+1,q^mn_t+i_2}=\left[
\begin{array}{c}
  {\bf b}_{i_2} \\
  1
\end{array}
\right]$
       \end{center}
       for any pair of distinct elements $i_1,i_2\in \{1,2,\cdots,q^{(m-1)t}\}$.
       Since ${\bf b}_{i_1}\neq {\bf b}_{i_2}$, $\left[
\begin{array}{c}
  {\bf b}_{i_1} \\
  1
\end{array}
\right]\notin \left[
\begin{array}{c}
  {\bf b}_{i_2} \\
  1
\end{array}
\right]R$. It turns out that ${\bf h}_1R\neq {\bf h}_2R$.
  \item Case 3. Suppose that ${\bf h}_1\in \mathcal{H}_1$ and ${\bf h}_2\in \mathcal{H}_2$.
  Since the every column of $Y_t$ contains at least one unit element of $R$, whereas the every column of $B_t$ does not contain unit elements of $R$, so
  ${\bf h}_1R\neq {\bf h}_2R$.
\end{itemize}
Hence the code generated by $Y_{t+1}$ is projective. Through inductive hypothesis, we complete the proof.
\end{proof}

Now, let us consider the Gray image of any Plotkin-optimal two-homogeneous weight regular projective linear code over $R$. First, we recall the parameters of SU1-type two-weight linear codes over $\F_q$.

\begin{prop}{\rm \cite[{\rm Example SU1}]{CK}}
Let $\F_q$ be a finite field. Then the parameters of SU1-type two-weight codes over $\F_q$ are as follows
$$n=\frac{q^\ell-q^s}{q-1}, k=\ell, \omega_1=q^{\ell-1}-q^{s-1},\omega_2=q^{\ell-1},
A_{\omega_1}=q^\ell-q^{\ell-s},A_{\omega_2}=q^{\ell-s}-1,$$
where $\ell>1$ and $1\leq s\leq \ell-1$.
\end{prop}

Let $C_1$ be the code obtained in Theorem \ref{thm two Lee weight Y_k over Z_2^k}.
By the Gray map in Subsection 2.3, we know that $\Phi(C_1)$ is a $q$-ary code with parameters
$$n=\frac{q^k-q^{k-t}}{q-1},~|\Phi(C_1)|=q^k,~\omega_1=q^{k-1}-q^{k-t-1},
~\omega_2=q^{k-1},~A_{\omega_1}=q^k-q^t,~A_{\omega_2}=q^t-1,$$
where $1\leq t\leq k_0$ and $k=mk_0+(m-1)k_1+\cdots+k_{m-1}$.
Hence $\Phi(C)$ has the same parameters and weight distribution as some two-weight
linear codes of type SU1. In this case, $\ell=mk_0+(m-1)k_1+\cdots+k_{m-1}$ and $s=\ell-t$. Note that $\Phi(C)$ is not necessarily linear, so sometimes they are not equivalent. It is worth mentioning that two-weight linear codes of type SU1 are obtained from finite geometry, however, we obtain them by their generator matrices in this paper. Hence our construction method is different.

\section{Applications in secret sharing schemes and graph theory}
As mentioned before, few-weight linear codes have applications in secret sharing
schemes, association schemes and authentication codes  \cite{SS-Ding,ass-Calderbank,ACM-Shamir,SSS-Blakley,authentication,IT-ding-1}. In particular, projective two-weight and three-weight codes are very precious as they are closely related to some combinatorial objects such as SRGs, SWRGs, difference sets and finite geometries \cite{CK,Shi-GC,IT-01,Shi-S-Sole,DCC-shi-sole,DCC-2-SRGs,FFA-2-weight}. Linear codes with few weights were investigated in \cite{IT-zhu,FFA-zhu-liao,AMC-zhu-liao,DM-cao,ccds-Cao-1,Liu-IT,Hyun-4, bocong-1,Hyun-1,Hyun-2,Hyun-3} and many other papers. Here, we study these applications of codes we constructed.

\subsection{Secret sharing schemes}
The concept of secret sharing schemes was first proposed by Blakley \cite{SSS-Blakley} and Shamir \cite{ACM-Shamir} in 1979. A general introduction to secret sharing schemes can be found for instance in Stinson's survey paper \cite{SSS-DCC-survey}.

The sets of participants which are capable of recovering the secret $S$ are called {\em access sets}. The {\em access structure} of a secret sharing scheme is defined to be the set of all access sets. A {\em minimal access set} is a group of participants who can recover the secret with their shares, but any of its proper subgroups can not do so. Furthermore, if a participant is contained in every minimal access set in the scheme, then it is a dictatorial participant.

The access structure of the secret sharing scheme based on a linear code is very complex in general, but it has been shown that for linear codes with all nonzero codewords minimal, their dual codes can be used to construct secret sharing schemes with nice access structures \cite{IT-yuan-ding,SS-Ding}. Now we recall the definition of minimal codewords.

The {\em support} of a vector ${\bf c}=(c_1,c_2,\ldots,c_n)\in \F_q^n$ is
defined by
$${\rm supp}({\bf c})=\{i~|~c_i\neq 0,~ 1\leq i\leq n\}.$$
We say that a vector ${\bf c}_1$ covers a vector ${\bf c}_2$ if ${\rm supp}({\bf c}_2)\subsetneq {\rm supp}({\bf c}_1)$. A {\em minimal codeword} of a linear code $C$ is a nonzero codeword that does not cover any other nonzero codeword of $C$. In general, it is a tough task to determine the minimal codewords of a given linear code. In some special cases, the Ashikhmin-Barg lemma \cite{IT-minimal-code} is very useful in determining the minimal codewords.

\begin{lemma}{\rm (Ashikhmin-Barg Lemma)}
Let $\omega_{min}$ and $\omega_{max}$ be the minimal and maximal nonzero weights of a linear code $C$ over $\F_q$, respectively. If
$$\frac{\omega_{min}}{\omega_{min}}>\frac{q-1}{q},$$
then each nonzero codeword of $C$ is minimal.
\end{lemma}

\begin{theorem}
Let $C$ be a linear $[n;k_0,k_1,\ldots,k_{m-1}]$ code over $R$ obtained in Theorem \ref{thm two Lee weight Y_k over Z_2^k}. If $t\geq 2$ and $\Phi(C)$ is linear, then all nonzero codewords of $\Phi(C)$ are minimal.
\end{theorem}

\begin{proof}
By Theorem \ref{thm two Lee weight Y_k over Z_2^k}, we have
$$\frac{\omega_{min}}{\omega_{min}}=\frac{q^{k-t-1}(q^t-1)}{q^{k-1}}=1-\frac{1}{q^t}
>\frac{q-1}{q}~{\rm for}~t\geq 2.$$
Hence the result follows.
\end{proof}

\subsection{Strongly regular graphs}

In this subsection, we will study that the relation between two-homogeneous codes we
construct and SRGs. First, we recall the definition of SRGs.

\begin{definition}
A simple graph on $N$ vertices is called $strongly\ regular\ graph$ with parameters
$(N,K, \lambda,\mu)$ if\\
(1) each vertex is adjacent to $K$ vertices,\\
(2) for each pair of adjacent vertices there are $\lambda$ vertices adjacent to both,\\
(3) for each pair of non-adjacent vertices there are $\mu$ vertices adjacent to both.
\end{definition}

The $spectrum$ of a graph $\Gamma$ is the set of eigenvalues of its adjacency matrix.
We will call an eigenvalue of $\Gamma$ $restricted$ if it has an eigenvector orthogonal to the all-ones vector $\textbf 1$. Note that for a $K$-regular connected graph, the restricted eigenvalues are simply the eigenvalues different from $K$ \cite{book-SRGs}.

 Byrne, Greferath and Honold \cite{DCC-2-SRGs} showed that the coset graphs of regular projective two-normalized homogeneous weight codes over finite chain rings are SRGs. Note that we call the unique homogeneous weight of average value $\gamma=1$ the {\em normalized}
homogeneous weight.

\begin{theorem}{\rm \cite[Theorem 5.5]{DCC-2-SRGs}}\label{DCC-srg}
Let $C$ be a regular projective two-weight code over a finite chain ring $R$ with normalized homogeneous weights $\omega_1$ and $\omega_2$. Then the graph $\Gamma:=(C,E)$ with vertex set $C$ and edge set $E:=\{\{x,y\}~|~{\bf x,y}\in C~ with~{\rm wt_{hom}}({\bf x} -{\bf y})=\omega_1 \}$ is strongly regular. Its parameters are $(N,K,\lambda,\mu)$, where
\begin{align*}
  N & =|C|, \\
  K & =\frac{(n-\omega_2)|C|+\omega_2}{\omega_1-\omega_2},\\
  \lambda&=\frac{nK\left[1-\left(1-\frac{\omega_1}{n}\right)^2\right]+\omega_2(1-K)}
  {\omega_1-\omega_2},\\
  \mu&=\frac{nK\left[1-\left(1-\frac{\omega_1}{n}\right) \left(1-\frac{\omega_2}{n}\right)\right]-\omega_2K}
  {\omega_1-\omega_2}.
\end{align*}
\end{theorem}

If we consider the normalized homogeneous weight, then the code in Theorem \ref{thm two Lee weight Y_k over Z_2^k} is a two-normalized homogeneous weight regular projective linear $[n;k_0,k_1,\ldots,k_{m-1}]_q$ code over $R$, where $n=\frac{q^k-q^{k-t}}{q^m-q^{m-1}}$ and $k=mk_0+(m-1)k_1+\cdots+k_{m-1}$. Moreover, the normalized homogeneous weight distribution is as follows.
$$
\begin{array}{ll}
\hline
{\rm Homogeneous\ weight}\ \omega &  {\rm Multiplicity}  \\ \hline
0 & 1 \\
\frac{q^{k-1}-q^{k-t-1}}{q^{m-1}-q^{m-2}} & q^k-q^t  \\
\frac{q^{k-1}}{q^{m-1}-q^{m-2}} & q^t-1\\ \hline
\end{array}
$$
Hence we have the following theorem.

\begin{theorem}
Let $C$ be the linear code above. Then the graph $\Gamma:=(C,E)$ with vertex set $C$ and edge set $E:=\{\{x,y\}~|~{\bf x,y}\in C~ with~{\rm wt_{hom}}({\bf x} -{\bf y})=\omega_1 \}$ is an SRG with parameters $(q^k,q^k-q^t,q^{k}-2q^t,q^k-q^t)$, which has unrestricted eigenvalues $0$ and $-q^t$.
\end{theorem}

\begin{proof}
The parameters of the graph $\Gamma$ is straightforward by Theorem \ref{DCC-srg}. By \cite[Section 1.1.1]{book-SRGs}, the restricted eigenvalues of $\Gamma$ are found as the roots of $x^2+(\mu-\lambda)x+(\mu-K)=0$. Hence the result follows.
\end{proof}

\begin{remark}
According to \cite[Section 1.1.3]{book-SRGs}, the graph $\Gamma$ is also a complete multipartite graph of type $K_{a\times m}$, where $a=q^t$ and $m=q^{k-t}$.
\end{remark}

Shi et al. \cite{Shi-S-Sole,IT-01,Shi-GC} showed that the coset graphs of the dual codes of regular projective two-homogeneous weight codes over finite chain rings are SRGs.
We recall some basic knowledge.
Let $[{\bf h}_1~|~{\bf h}_2~|~\cdots~|~{\bf h}_n]$ be a parity-check matrix of $C$. The vertex set of {\em the syndrome graph} $\Gamma(C)$ is $V=\{H{\bf x}~|~{\bf x}\in R^n\}$, the column space of $H$, which is isomorphic to the dual code $C^\perp$. Two vertexes $H{\bf x}$, $H{\bf y}$ are adjacent in $\Gamma(C)$ if they differ by a unit-multiple of a column of $H$:
$$H{\bf x}\sim H{\bf y}\Longleftrightarrow H({\bf x}-{\bf y})=u{\bf h}_i,$$
for some $1\leq i\leq n$ and $u\in R^\times.$ It is obvious that $\Gamma(C)$
is the Cayley graph of $V$ corresponding to the generating set $S=\{u{\bf h}_i~|~u\in R^\times,~1\leq i\leq n\}$. Since ${\bf 0}\notin S$ and $-S\subseteq S$, $\Gamma(C)$ is simple. Moreover, $\Gamma(C)$ is regular of degree $|S|=(q^m-q^{m-1})n$ and has $|C^\perp|$ vertices.

\begin{theorem}{\rm \cite{Shi-S-Sole,IT-01,Shi-GC}}\label{thm-srg-shi}
Suppose that $C$ is a regular, projective linear two-weight code over a finite chain ring $R$ with nonzero homogeneous weights $\omega_1$ and $\omega_2$. Then the syndrome graph $\Gamma(C^\perp)$ is an SRG with $(q^m-q^{m-1})n$ and the restricted eigenvalues $K(\omega_1)$ and $K(\omega_2)$, where $K(x)=(q^m-q^{m-1})n-qx$.
\end{theorem}

By Theorems \ref{thm two Lee weight Y_k over Z_2^k} and \ref{thm-srg-shi}, we have the following theorem.

\begin{theorem}
Let $C$ be a linear $[n;k_0,k_1,\ldots,k_{m-1}]$ code over $R$ obtained in Theorem \ref{thm two Lee weight Y_k over Z_2^k}. Then the syndrome graph $\Gamma(C^\perp)$ of $C^\perp$ is an SRG with parameters $(q^k,q^k-q^{k-t},q^{k}-2q^{k-t},q^k-q^{k-t})$, which has unrestricted eigenvalues $0$ and $-q^{k-t}$.
\end{theorem}

\begin{proof}
By the above discussion,
$N=|(C^\perp)^\perp|=|C|=q^k$ and $K=|S|=(q^m-q^{m-1})n=q^k-q^{k-t}$. By Theorem \ref{thm-srg-shi}, the syndrome graph $\Gamma(C^\perp)$ of $C^\perp$ is an SRG with unrestricted eigenvalues $0$ and $-q^{k-t}$. By \cite[Section 1.1.1]{book-SRGs}, the restricted eigenvalues of $\Gamma$ are found as the roots of $x^2+(\mu-\lambda)x+(\mu-K)=0$. Hence the result follows.
\end{proof}

\subsection{Strongly walk-regular graphs}

In this subsection, we apply Theorem \ref{thm-1} to some known three-homogeneous weight regular projective codes over finite chain rings. Several families of three-homogeneous weight regular projective codes are constructed. As an application, we construct several families of Strongly walk-regular graphs (SWRGs). SWRGs were introduced in \cite{JCTA-SWRG} as a natural generalization of SRGs, where paths of length $2$ are replaced by paths of length $\ell\geq 2$. Specifically, a graph is $\ell$-SWRG if there are three integers $(\lambda,\mu,\nu)$ such that the number of paths of length $\ell$ between any two vertices $x$ and $y$ is
\begin{itemize}
  \item $\lambda$ if $x$ and $y$ are connected;
  \item $\mu$ if $x$ and $y$ are disconnected;
  \item $\nu$ if $x=y$.
\end{itemize}

A set $\Omega\subseteq R^k$ is called an {\em $\ell$-sum set} if it is stable by scalar multiplication by units and there are constants $\sigma_0$ and $\sigma_1$ such that each non-zero $h\in R^k$ can be written as $h=\sum_{i=1}^\ell x_i$ with $x_i\in \Omega$ exactly $\sigma_0$ times if $h\in \Omega$ and $\sigma_1$ times if $h\in R^k\setminus \Omega$. Note that $\ell$-sum sets are a natural generalization of partial difference
sets and triple sum sets (TSSs).
Let $\overline{R^k}$ denote the set of all regular vectors in $R^k$, $i.e.$,
$$\overline{R^k}=\{{\bf x}\in R^k~|~\{{\bf x}\cdot {\bf y}~|~{\bf y}\in R^k\}=R\}.$$
Assume that $\Omega\subseteq \overline{R^k}$ such that ${\bf 0}\notin \Omega$ and it is stable by scalar multiplication by units. Let $H$ be a $k\times n$ matrix whose columns are all projectively nonequivalent elements of $\Omega$. Then we denote by $C(\Omega)$ the regular and projective code of length $\frac{|\Omega|}{|R^\times|}$ with the parity-check matrix $H$. In \cite[Theorem 2]{Shi-GC}, it was shown that $\Omega$ is an $\ell$-sum set if and only if $\Gamma(C(\Omega))$ is an $\ell$-SWRG. This paper only focuses on the case of $\ell=3$.

\begin{theorem}{\rm \cite[Theorem 12]{Shi-GC}}\label{GC-shi}
Assume that $C(\Omega)^\perp$ is of length $n$ and has three nonzero homogeneous weights $\omega_1<\omega_2<\omega_3$. Then $\Omega$ is a TSS if and only if $\omega_1+\omega_2+\omega_3=3n(q^{m-1}-q^{m-2})$.
\end{theorem}

Let $C$ be an $N$-homogeneous weight regular projective linear code over $R$ so that $C$ contains at least a codeword with homogeneous weight $n(q^{m-1}-q^{m-2})$. Applying Theorem \ref{thm-1} to the code $C$, one can construct a family of $N$-homogeneous weight regular projective linear codes over $R$. Then we have the following theorem.

\begin{theorem}\label{thm-3-TSS}
Let $\Omega$ be a $3$-sum set such that $C(\Omega)^\perp$ is a regular projective linear code of length $n$ over $R$ with three nonzero homogeneous weights $\omega_1<\omega_2=n(q^{m-1}-q^{m-2})<\omega_3$ and $\omega_1+\omega_3=2n(q^{m-1}-q^{m-2})$. Let $\Omega'=\{({\bf a},b)~|~{\bf a}\in \Omega,~b\in \langle \theta^j\rangle\}$. Then $\Omega'$ is also a $3$-sum set, i.e., $\Gamma(C(\Omega'))$ is a $3$-SWRG.
\end{theorem}

\begin{proof}
By Theorem \ref{thm-1}, $C(\Omega')^\perp$ is a regular projective linear code of length $n'=q^{m-j}n$ with three nonzero homogeneous weights $\omega'_1<\omega'_2<\omega'_3$ and $\omega_i'=q^{m-j}\omega_i$ for $1\leq i\leq 3$. Hence $\omega_1'+\omega_2'+\omega_3'=3n'(q^{m-1}-q^{m-2})$. It follows from Theorem \ref{GC-shi} that $\Omega'$ is a TSS. By \cite[Theorem 2]{Shi-GC}, $\Gamma(C(\Omega'))$ is a $3$-SWRG.
\end{proof}

\begin{remark}
By Theorem \ref{thm-3-TSS}, we construct an family of TSSs or 3-SWRGs from a known three-weight regular projective code over finite chain rings whose weights satisfy a certain equation. There were many classification of three-weights regular projective codes of short length over finite chain rings, such as $\F_2$ \cite{DCC-Kiermaier,DCC-shi-sole}, $\F_3$ \cite{DCC-Kiermaier,DCC-shi-sole}, $\Z_4$ \cite{Shi-GC} and $\F_2+u\F_2$ with $u^2=0$ \cite{Shi-GC}.
\end{remark}

\begin{example}
Let $\Omega=\{100,300,010,030,102,302,132,312,130,310,210,230\}\subseteq \Z_4^3$, which is stable by scalar multiplication by units. Then $C(\Omega)$ is the regular and projective code of length $\frac{|\Omega|}{|\Z_4^\times|}$ with the parity-check matrix
$$H=\left[\begin{array}{c}
101112\\
010331\\
002200
\end{array}\right].$$
It can be checked that $C(\Omega)^\perp$ is a $[6;2,1]_2$ code over $\Z_4$ and has Lee weight distribution
$$[ \langle0, 1\rangle, \langle4, 6\rangle, \langle6, 16\rangle, \langle8, 9\rangle ].$$
By Theorem \ref{GC-shi}, $\Omega$ is a TSS.  Let $\Omega'=\{1000,3000,0100,0300, 1020,3020,1320,3120,1300,\\3100,2100,2300, 1002,3002,0102,0302,1022,3022,1322,3122,
1302,3102,2102,2302\}$. By Theorem \ref{thm-1}, $C(\Omega')^\perp$ is a regular and projective code of length $12$ with the parity-check matrix
$$H'=\left[\begin{array}{c|c}
101112& 101112\\
010331& 010331 \\
002200& 002200\\
000000&222222
\end{array}\right].$$
 By Theorem \ref{thm-3-TSS}, we know that $\Omega'$ is a TSS and $\Gamma(C(\Omega'))$ is a $3$-SWRG.

\end{example}

\section{Conclusion}
In this paper, we have characterized all possible parameters of Plotkin-optimal two-homogeneous weight regular projective linear codes over finite chain rings, as well as their weight distributions. We show the existence of codes with these parameters by constructing an infinite family of two-homogeneous weight codes. The parameters of their Gray images have the same weight distribution as that of the two-weight codes of type SU1 in the sense of Calderbank and Kantor \cite{CK}. Furthermore, we also have constructed three-homogeneous weight regular projective codes over finite chain rings combined with some known results. Finally, we have studied applications of our constructed codes in secret sharing schemes and graph theory.\\

\noindent{\bf Data availability:} No data was used for the research described in the paper.\\

\noindent{\bf Conflict of Interest:} The authors have no conflicts of interest to declare that are relevant to the content of this paper.\\

\noindent{\bf Acknowledgement:} This research is supported by the National Natural Science Foundation of China (12071001 and 12201170) and the Natural Science Foundation of Anhui Province (2108085QA03).\\

\end{document}